\documentclass[letterpaper, 10 pt, conference]{ieeeconf}  

\IEEEoverridecommandlockouts                              

\usepackage{graphicx, enumerate}
\usepackage{cite,array, bm}
\usepackage{dsfont, epstopdf}
\usepackage{mathrsfs}
\usepackage{amssymb}
\usepackage[cmex10]{amsmath}
\usepackage{subeqnarray, stmaryrd}
\usepackage{cases}
\usepackage[noend]{algpseudocode}
\usepackage{algorithmicx,algorithm}

\newtheorem{theorem}{Theorem}

\newtheorem{lemma}{Lemma}
\newtheorem{proposition}{Proposition}
\newtheorem{definition}{Definition}

\DeclareMathOperator{\Id}{Id}

\DeclareMathOperator{\enab}{enab}
\DeclareMathOperator{\dom}{dom}

\DeclareMathOperator{\cp}{ConPre}

\begin{document}

\title{\LARGE \bf Logarithmic Quantization based Symbolic Abstractions for Nonlinear Control Systems
\thanks{This work was supported by the H2020 ERC Starting Grant BUCOPHSYS, the EU H2020 Co4Robots Project, the Swedish Foundation for Strategic Research (SSF), the Swedish Research Council (VR) and the Knut och Alice Wallenberg Foundation (KAW).}
}

\author{Wei~Ren, and Dimos V. Dimarogonas 
\thanks{W. Ren and D. Dimarogonas are with Division of Decision and Control Systems, EECS, KTH Royal Institute of Technology, SE-10044, Stockholm, Sweden.
Email: \texttt{\small weire@kth.se}, \texttt{\small dimos@kth.se}.}
}

\maketitle

\begin{abstract}
This paper studies symbolic abstractions for nonlinear control systems using logarithmic quantization. With a logarithmic quantizer, we approximate the state and input sets, and then construct a novel discrete abstraction for nonlinear control systems. A feedback refinement relation between the constructed discrete abstraction and the original system is established. Using the constructed discrete abstraction, the safety controller synthesis problem is studied. With the discrete abstraction and the abstract specification, the existence of a safety controller is investigated, and the algorithm is proposed to compute the abstract controller. Finally, a numerical example is given to illustrate the obtained results.
\end{abstract}

\section{Introduction}
\label{sec-intro}

The use of discrete abstractions \cite{Milner1989communication, Tabuada2006linear} has gradually become a standard approach for the design of hybrid systems due to the following two main advantages. First, thanks to discrete abstractions of continuous dynamics, one can deal with controller synthesis problems efficiently via techniques developed in the fields of supervisory control \cite{Ramadge1987supervisory} or algorithmic game theory \cite{Ramadge1987modular}. Second, with an inclusion or equivalence relationship between the original system and the discrete abstraction, the synthesized controller is guaranteed to be correct by design, and thus formal verification is either not needed or can be reduced \cite{Girard2012controller}. To construct the discrete abstraction, the key is to find an equivalence relation on the state space of dynamic systems. Such an equivalence relation leads to a new system, which is on the quotient space and shares the interested properties with the original system.

In the literature on the construction of the discrete abstraction, the most commonly-used approach is based on (alternating) (bi-)simulation relations and their approximate variants in \cite{Pola2008approximately, Girard2010approximately}. The simulation relation and related concepts capture equivalences of dynamic systems in an exact or approximate setting. However, this type of relations results in the requirement of exact information of the original system to obtain the refined controller, and a huge computational complexity for the abstract controller due to its abstraction refinement. As a result, the feedback refinement relation was proposed in \cite{Reissig2014feedback}, and provides an alternative to connect the discrete abstraction and the original system. With a feedback refinement relation, the abstract controller can be connected to the original system via a static quantizer \cite{Reissig2017feedback}. Some salient results can be found; see \cite{Khaled2016symbolic, Meyer2018compositional}.

On the other hand, due to time-invariant quantization regions and the resulting simple structures \cite{Ren2018quantized}, a static quantizer is applied in the construction of discrete abstractions \cite{Pola2008approximately, Girard2010approximately}. The uniform quantizer, which is a static quantizer with uniform time or space partitions \cite{Delchamps1990stabilizing}, is commonly used in approximations of both the state and input sets. Since the uniform quantization partitions the state set with equal distance, a huge computational complexity may be needed to compute the discrete abstraction \cite{Pola2008approximately, Girard2010approximately}. To reduce the computational complexity, a coarse quantizer \cite{Fu2005sector, Coutinho2010input, Elia2001stabilization} can be instead applied such that the state or input space can be partitioned with different distance, and this is the main motivation of this paper. Using the coarse quantizers like logarithmic quantizer and hysteresis quantizer, the approximate bisimulation is not valid any more, and thus the feedback refinement relation is applied.

In this paper, we study the discrete abstraction and controller synthesis of nonlinear control systems via logarithmic quantization. Using the logarithmic quantizer of \cite{Fu2005sector}, which is a coarse quantizer, the state and input sets are approximated, and then a novel discrete abstraction is constructed. According to the constructed discrete abstraction, the safety controller synthesis is studied via abstract specification, which is obtained via the logarithmic quantizer. A numerical example is given to demonstrate the obtained results. The main contributions of this paper are three-fold. To begin with, logarithmic quantization based discrete abstraction is first proposed, which provides an alternative approach to approximate the state and input sets. In addition, since the logarithmic quantization is coarser than the uniform quantization, the computational complexity of the obtained discrete abstraction is reduced greatly. Second, using the feedback refinement relation proposed in \cite{Reissig2017feedback}, abstract specification is constructed via logarithmic quantization, and further used in controller synthesis. Third, using the obtained abstraction and the abstract specification, the safety controller synthesis is investigated for the original system, and an algorithm is proposed to construct the safety controller.

\section{Nonlinear Control Systems}
\label{sec-nonlinear}

\subsection{Notations}

We denote $\mathbb{R}:=(-\infty, +\infty)$; $\mathbb{R}^{+}_{0}:=[0, +\infty)$; $\mathbb{R}^{+}:=(0, +\infty)$; $\mathbb{N}:=\{0, 1, \ldots\}$; $\mathbb{N}^{+}:=\{1, 2, \ldots\}$. $\|\cdot\|$ represents the infinite vector norm. Given $a, b\in\mathbb{R}\cup\{\pm\infty\}$ with $a\leq b$, we denote by $[a, b]$ a closed interval. Given $a, b\in(\mathbb{R}\cup\{\pm\infty\})^{n}$, we define the relations $<, >, \leq, \geq$ on $a, b$ component-wise. Given $x\in\mathbb{R}^{n}$, $x_{i}$ denotes the $i$-th element of $x$, $|x_{i}|$ denotes the absolute value of $x_{i}$, and $|x|=(|x_{i}|, \ldots, |x_{n}|)$. A cell $\llbracket a, b\rrbracket$ is the closed set $\{x\in\mathbb{R}^{n}|a_{i}\leq x_{i}\leq b_{i}\}$. Given two sets $A, B\subset\mathbb{R}^{n}$ with $A\subseteq B$, denote by $\Id_{A}: A\hookrightarrow B$ the natural inclusion map from $a\in A$ to $\Id(a)=a\in B$. A relation $\mathcal{R}\subset A\times B$ with the map $\mathcal{R}: A\rightarrow2^{B}$ defined by $b\in\mathcal{R}(a)$ if and only if $(a, b)\in\mathcal{R}$. $\mathcal{R}^{-1}$ denotes the inverse relation of $\mathcal{R}$, i.e. $\mathcal{R}^{-1}:=\{(b, a)\in B\times A: (a, b)\in\mathcal{R}\}$. Given a set $A$, $A^{[0, t)}$ denotes the set of all the signals, which take values in $A$ and are defined on intervals of the form $[0, t)$; $A^{\infty}=\bigcup_{t\in\mathbb{N}}A^{[0, t)}$.

\subsection{Nonlinear Control Systems}

The class of nonlinear control systems considered in this paper is introduced in the following definition.

\begin{definition}[\cite{Pola2008approximately}]
\label{def-1}
A \textit{control system} $\Sigma$ is a quadruple $\Sigma=(\mathbb{R}^{n}, U, \mathcal{U}, f)$, where,
\begin{itemize}
\item $\mathbb{R}^{n}$ is the state set;
\item $U\subseteq\mathbb{R}^{m}$ is the input set;
\item $\mathcal{U}$ is a subset of all piecewise continuous functions from the interval $(a, b)\subset\mathbb{R}$ to $U$, with $a<0, b>0$;
\item $f: \mathbb{R}^{n}\times U\rightarrow\mathbb{R}^{n}$ is a continuous map satisfying the following Lipschitz assumption: there exists a constant $L\in\mathbb{R}^{+}$ such that for all $x, y\in\mathbb{R}^{n}$ and all $u\in U$, we have $\|f(x, u)-f(y, u)\|\leq L\|x-y\|$.
\end{itemize}
\end{definition}

A curve $\xi: (a, b)\rightarrow\mathbb{R}^{n}$ is said to be a \textit{trajectory} of $\Sigma$, if there exists $u\in U$ such that $\dot{\xi}(t)=f(\xi(t), u(t))$ for almost all $t\in(a, b)$. Different from the trajectory defined above over the open domain, we refer to the trajectory $\mathbf{x}: [0, \tau]\rightarrow\mathbb{R}^{n}$ defined on a closed domain $[0, \tau]$ with $\tau\in\mathbb{R}^{+}$ such that $\mathbf{x}=\xi|_{[0, \tau]}$. Denote by $\mathbf{x}(t, x, u)$ the point reached at time $t\in(a, b)$ under the input $u$ from the initial condition $x$. Such a point is uniquely determined, since the assumptions on $f$ ensure the existence and uniqueness of the trajectory.

\section{Feedback Refinement Relation}
\label{sec-feedback}

In this section, we introduce the notion of feedback refinement relation upon which the following results rely. To begin with, the class of transition systems is introduced.

\begin{definition}[\cite{Girard2012controller}]
\label{def-2}
A \textit{transition system} is a sextuple $T=(X, X^{0}, U, \Delta, Y, H)$, comprising of: (i) a set of states $X$; (ii) a set of initial states $X^{0}$; (iii) a set of inputs $U$; (iv) a transition relation $\Delta: X\times U\times X$; (v) an output set $Y$; (vi) an output map $H : X\rightarrow Y$.
\end{definition}

The transition $(x, u, x')\in\Delta$ is denoted by $x'\in\Delta(x, u)$, which means that the system can evolve from $x$ to $x'$ under the input $u$. An input $u\in U$ is said to belong to \textit{the set of enabled inputs} at $x\in X$, denoted by $\enab(x)$, if $\Delta(x, u)\neq\varnothing$. If $\enab(x)=\varnothing$, then $x\in X$ is said to be \textit{blocking}; otherwise, it is said to be \textit{non-blocking}. If all the states are non-blocking, the system $T$ is called to be \textit{non-blocking}.

Similar to approximate simulation relations and their variants in \cite{Pola2008approximately, Girard2012controller}, a feedback refinement relation between two transition systems $T_{1}$ and $T_{2}$ is introduce as follows.

\begin{definition}[\cite{Reissig2017feedback}]
\label{def-3}
Let $T_{i}=(X_{i}, X^{0}_{i}, U_{i}, \Delta_{i}, Y_{i}, H_{i})$ be two transition systems with $i\in\{1, 2\}$, and assume that $U_{2}\subseteq U_{1}$. A relation $\mathcal{F}\subseteq X_{1}\times X_{2}$ is a \textit{feedback refinement relation} from $T_{1}$ to $T_{2}$, if for all $(x_{1}, x_{2})\in\mathcal{F}$, (i) $U_{2}(x_{2})\subseteq U_{1}(x_{1})$; (ii) $u\in U_{2}(x_{2}), x'_{1}=\Delta_{1}(x_{1}, u)\Rightarrow\mathcal{F}(x'_{1})\subseteq\Delta_{2}(x_{2}, u)$, where $U_{i}(x):=\{u\in U_{i}: \Delta_{i}(x, u)\neq\varnothing\}$. Denote $T_{1}\preceq_{\mathcal{F}}T_{2}$ if $\mathcal{F}$ is a feedback refinement relation from $T_{1}$ to $T_{2}$.
\end{definition}

\section{Symbolic Model}
\label{sec-symbolic}

In this section, we work with the time-discretization of the control system $\Sigma$. Assume the sampling period is $\tau>0$, which is a design parameter. We define the time-discretization of the control system $\Sigma$ as the transition system $T_{\tau}(\Sigma):=(X_{1}, X^{0}_{1}, U_{1}, \Delta_{1}, Y_{1}, H_{1})$, where,
\begin{itemize}
\item the set of states is $X_{1}:=\mathbb{R}^{n}$;
\item the set of initial states is $X^{0}_{1}:=\mathbb{R}^{n}$;
\item the set of inputs is $U_{1}:=\{u\in\mathcal{U}|\mathbf{x}(t, x, u) \text{ is defined for all } x\in\mathbb{R}^{n}\}$;
\item the transition relation is given as follows: for $x\in X_{1}$ and $u\in U_{1}$, $x'=\Delta_{1}(x, u)$ if and only if $x'=\mathbf{x}(\tau, x, u)$;
\item the set of outputs is $Y_{1}:=\mathbb{R}^{n}$;
\item the output map is $H: X_{1}\hookrightarrow X_{1}$.
\end{itemize}

\subsection{Logarithmic Quantization based Approximation}
\label{subsec-quantizer}

To construct a discrete abstraction of a control system, the state and input sets need to be approximated first. To reduce the computational complexity, the following logarithmic quantizer is applied, which provides an alternative for the approximation of the state and input sets.

\begin{definition}[\cite{Fu2005sector, Coutinho2010input, Elia2001stabilization}]
\label{def-4}
A quantizer is called a \textit{logarithmic quantizer}, if it has the following form
\begin{align}
\label{eqn-2}
Q(z):=\left\{\begin{aligned}
&z_{i},  & & (1+\eta)^{-1}z_{i}<z\leq(1-\eta)^{-1}z_{i}; \\
&0, & & 0\leq z\leq(1+\eta)^{-1}d;  \\
& -Q(-z), & & z<0, \\
\end{aligned}\right.
\end{align}
where $z_{i}=\rho^{(1-i)}d$, $\rho=\frac{1-\eta}{1+\eta}$, $\eta\in(0, 1)$, $d>0$, and $i\in\mathbb{N}^{+}$.
\end{definition}

In Definition \ref{def-4}, the parameter $\rho\in(0, 1)$ is called the quantization density and $z_{\min}:=(1+\eta)^{-1}d$ is the size of the deadzone. For a quantized measurement $z_{i}>0$, the quantization region is $\hat{z}_{i}:=((1+\eta)^{-1}z_{i}, (1-\eta)^{-1}z_{i}]$. The quantization error $z-Q(z)$ can be written as (see \cite{Fu2009finite})
\begin{equation}
\label{eqn-3}
z-Q(z):=\Lambda(z)z, \quad \Lambda(z)\in[-\eta, \eta].
\end{equation}

Using the logarithmic quantizer \eqref{eqn-2}, the state set $\mathbb{R}^{n}$ is approximated by the sequence of embedded lattices $[\mathbb{R}^{n}]_{\eta}$:
\begin{align*}
[\mathbb{R}^{n}]_{\eta}&:=\left\{q\in\mathbb{R}^{n}: q_{i}=\pm\frac{\rho^{(1-k_{i})}d}{\sqrt{n}}, k_{i}\in\mathbb{N}^{+}, \right. \\
&\quad i\in\{1, \ldots, n\}\}\cup\{0\},
\end{align*}
where, $\rho=(1+\eta)^{-1}(1-\eta)$, $\eta\in(0, 1)$ is treated as a state space parameter, and $d>0$ is a fixed constant. We associate a quantizer $Q_{\eta}: \mathbb{R}^{n}\rightarrow[\mathbb{R}^{n}]_{\eta}$ such that $Q_{\eta}(x)=Q(x)$ if and only if for $x=(x_{1}, \ldots, x_{n})\in\mathbb{R}^{n}$ and $i\in\{1, \ldots, n\}$,
\begin{equation*}
(\sqrt{n}(1+\eta))^{-1}|q_{i}|\leq|x_{i}|\leq(\sqrt{n}(1-\eta))^{-1}|q_{i}|,
\end{equation*}
or
\begin{equation*}
-(\sqrt{n}(1+\eta))^{-1}d\leq x_{i}\leq(\sqrt{n}(1+\eta))^{-1}d.
\end{equation*}
As a result, from \eqref{eqn-3} and simple geometrical considerations, $\|x-Q_{\eta}(x)\|\leq\Lambda(x)\|x\|$ holds for all $x\in\mathbb{R}^{n}$, where $\Lambda(x)\in[-\eta, \eta]$. With the quantizer $Q_{\eta}$, the state set is partitioned as
\begin{align*}
\hat{X}&:=\bigcup_{q\in[\mathbb{R}^{n}]_{\eta}}\hat{q},
\end{align*}
where $\hat{q}$ is the quantization region corresponding to the quantized measurement $q\in[\mathbb{R}^{n}]_{\eta}$.

In the following, the approximation of the input set $U_{1}$ of $T_{\tau}(\Sigma)$ is presented; see also \cite{Pola2008approximately} for a similar mechanism. We approximate $U_{1}$ by means of the set:
\begin{equation}
\label{eqn-4}
U_{2}:=\bigcup_{q\in[\mathbb{R}^{n}]_{\eta}}U_{2}(\hat{q}),
\end{equation}
where $U_{2}(\hat{q})$ captures the set of inputs that can be applied at the symbolic state $\hat{q}\in\hat{X}$. $U_{2}(\hat{q})$ is defined based on the reachable sets. Starting from a state $q\in[\mathbb{R}^{n}]_{\eta}$ (thus $q\in X_{1}$), the set of reachable states of $T_{\tau}(\Sigma)$ is obtained below.
\begin{equation*}
\mathfrak{R}(\tau, q):=\{x'\in X_{1}: \mathbf{x}(\tau, q, u)=x', u\in U_{1}\},
\end{equation*}
which is well-defined from the definition of the input set $U_{1}$.

The reachable set $\mathfrak{R}(\tau, q)$ is approximated as follows. Given any $\mu\in\mathbb{R}^{+}$, consider the following set
\begin{equation*}
\mathcal{Z}_{\mu}(\tau, q):=\{y\in[\mathbb{R}^{n}]_{\mu}: \exists z\in\mathfrak{R}(\tau, q) \text{ s.t. } y=Q_{\mu}(z)\}.
\end{equation*}
Here, $\mu$ is a design parameter, whose choice is not related to $\eta$. Define the function $\phi: \mathcal{Z}_{\mu}(\tau, q)\rightarrow U_{1}$, which means that, for any $y\in\mathcal{Z}_{\mu}(\tau, q)$, there exists an input $u_{1}=\phi(y)\in U_{1}$ such that $y=Q_{\mu}(\mathbf{x}(\tau, q, u_{1}))$. Note that the function $\phi$ is not unique. Thus, the set $U_{2}(\hat{q})$ in \eqref{eqn-4} can be defined by $U_{2}(\hat{q}):=\phi(\mathcal{Z}_{\mu}(\tau, q))$. Since the set $U_{2}(\hat{q})$ is the image through the map $\phi$ of a countable set, we have that $U_{2}(\hat{q})$ is countable, which implies that $U_{2}$ as defined in \eqref{eqn-4} is countable. As a result, the set $U_{2}$ approximates the set $U_{1}$ in the following way: given any $q\in[\mathbb{R}^{n}]_{\eta}$, for any $u_{1}\in U_{1}$, there exists $u_{2}\in U_{2}(\hat{q})$ such that $Q_{\mu}(\mathbf{x}(\tau, q, u_{1}))=Q_{\mu}(\mathbf{x}(\tau, q, u_{2}))$. That is, $\mathbf{x}(\tau, q, u_{1})$ and $\mathbf{x}(\tau, q, u_{2})$ are in the same quantization region.

In contrast to the uniform quantization of the state and input sets as in \cite{Pola2008approximately, Girard2012controller}, the logarithmic partition proposed here significantly reduces the computation complexity of the developed abstraction; see Section \ref{sec-example}.

\subsection{Symbolic Abstraction}
\label{subsec-symbolic}

With the partitions of the state and input sets, the symbolic abstraction of the system $T_{\tau}(\Sigma)$ is described in this subsection. The developed symbolic abstraction is a transition system $T_{\tau, \eta, \mu}(\Sigma)=(X_{2}, X^{0}_{2}, U_{2}, \Delta_{2}, Y_{2}, H_{2})$, where,
\begin{itemize}
\item the set of states is $X_{2}=\hat{X}$;
\item the set of initial states is $X^{0}_{2}=\hat{X}$;
\item the set of inputs is $U_{2}=\bigcup_{q\in [\mathbb{R}^{n}]_{\eta}}U_{2}(q)$;
\item the transition relation is given as follows: for $\hat{q}_{1}, \hat{q}_{2}\in X_{2}$ and $u\in U_{2}$, $\hat{q}_{2}\in\Delta_{2}(\hat{q}_{1}, u)$ if and only if
\begin{align}
\label{eqn-5}
\hat{q}_{2}&\cap\left(\mathbf{x}(\tau, q_{1}, u)+\llbracket-\theta e^{L\tau}\bar{q}_{1}, \theta e^{L\tau}\bar{q}_{1}\rrbracket\right)\neq\varnothing,
\end{align}
where $\theta:=\eta(1-\eta)^{-1}$,  $\bar{q}_{1}:=|q_{1}|+E_{q_{1}}$, $E_{q_{1}}\in\mathbb{R}^{n}$ is a vector whose the components are 1 if the corresponding components of $q_{1}$ are 0; and zero otherwise, and $L>0$ is the Lipschitz constant of the function $f$;
\item the set of outputs is $Y_{2}=\mathbb{R}^{n}$;
\item the output map is $H_{2}: X_{2}\hookrightarrow X_{2}$.
\end{itemize}

In the construction of the abstraction $T_{\tau, \eta, \mu}(\Sigma)$, the technique applied in \eqref{eqn-5} is similar to those in \cite{Meyer2018compositional, Reissig2011computing, Reissig2017feedback}, where the overapproximation of successors of states is applied. $\theta e^{L\tau}\bar{q}_{1}$ in \eqref{eqn-5} plays the same role as the growth bound in \cite{Reissig2017feedback}. Since the logarithmic quantizer is implemented here, the components of $X_{2}$ are the quantization region related to the quantized measurements. Hence, the developed abstraction extends those in previous works \cite{Pola2008approximately, Reissig2017feedback}, and provides an alternative for the abstraction construction. On the other hand, due to the logarithmic quantizer, the quantization errors are not bounded. Hence, the abstraction $T_{\tau, \eta, \mu}(\Sigma)$ and the system $T_{\tau}(\Sigma)$ do not satisfy the approximate bisimulation relation; see \cite{Pola2008approximately}. To deal with this issue, a feedback refinement relation is applied to connect $T_{\tau}(\Sigma)$ with $T_{\tau, \eta, \mu}(\Sigma)$.

\begin{theorem}
\label{thm-1}
Consider the control system $\Sigma$ with the time and state space sampling parameters $\tau, \eta, \mu\in\mathbb{R}^{+}$. Let  the map $\mathcal{F}: X_{1}\rightarrow X_{2}$ be given by $\mathcal{F}(x)=\hat{q}$ if and only if $x\in\hat{q}$. Then $T_{\tau}(\Sigma)\preceq_{\mathcal{F}}T_{\tau, \eta, \mu}(\Sigma)$.
\end{theorem}

\begin{proof}
Following from the definitions of $T_{\tau}(\Sigma)$ and $T_{\tau, \eta, \mu}(\Sigma)$, one has that $U_{2}\subseteq U_{1}$. Let $(x_{1}, \hat{q}_{1})\in\mathcal{F}$ with $x_{1}\in X_{1}$ and $\hat{q}_{1}\in X_{2}$, and we have that $x_{1}\in\hat{q}_{1}$. For each $u\in U_{2}(\hat{q}_{1})$, we obtain that $u\in U_{2}(\hat{q}_{1})\subseteq U_{2}\subseteq U_{1}$. In addition, $\Delta_{2}(\hat{q}_{1}, u)\neq\varnothing$ holds from the definition of $U_{2}(\hat{q}_{1})$. If $\Delta_{1}(x_{1}, u)=\varnothing$, then we have that $u\notin U_{1}$, which is a contradiction. As a result, $\Delta_{1}(x_{1}, u)\neq\varnothing$ and $u\in U_{1}(x_{1})$. We thus conclude that $U_{2}(\hat{q}_{1})\subseteq U_{1}(x_{1})$.

Given $\hat{q}_{1}, \hat{q}_{2}\in X_{2}$ and $u\in U_{2}(\hat{q}_{1})$, define $x_{2}:=\Delta_{1}(x_{1}, u)$, and thus $x_{1}\in\hat{q}_{1}$ holds from $(x_{1}, \hat{q}_{1})\in\mathcal{F}$, combining which with \eqref{eqn-2} yields that $\|x_{1}-q_{1}\|\leq\theta\|q_{1}\|$. If $\Delta_{1}(x_{1}, u)\cap\hat{q}_{2}\neq\varnothing$, there exists $x_{2}:=\mathbf{x}(\tau, x_{1}, u)\in X_{1}$ such that $x_{2}\in\hat{q}_{2}$ holds from \eqref{eqn-5}. From the Lipschitz property of the function $f$, one has $\|\mathbf{x}(\tau, x_{1}, u)-\mathbf{x}(\tau, q_{1}, u)\|\leq e^{L\tau}\|x_{1}-q_{1}\|\leq\theta e^{L\tau}\|q_{1}\|$, which implies that $\hat{q}_{2}\cap(\mathbf{x}(\tau, q_{1}, u)+\llbracket-\theta e^{L\tau}\bar{q}_{1}, \theta e^{L\tau}\bar{q}_{1}\rrbracket)\neq\varnothing$. Hence, $\hat{q}_{2}\in X_{2}$ holds from the construction of the abstraction $T_{\tau, \eta, \mu}(\Sigma)$, which in turn completes the proof.
\end{proof}

In the proof of Theorem \ref{thm-1}, $\Delta_{1}(x_{1}, u)\cap\hat{q}_{2}\neq\varnothing$ holds due to the unbounded state set studied in this paper. If the state set is bounded as in practical systems, we can impose an additional requirement such that $\Delta_{2}(\hat{q}_{1}, u)=\varnothing$ if $\hat{q}_{1}$ does not belong to the state set. Therefore, the feedback refinement relation is still valid in this case; see also \cite{Reissig2017feedback, Meyer2018compositional}. Since the logarithmic quantization is coarse and may lead to large approximation error, we can reduce the approximation error by applying  logarithmic quantization to the components of $X_{2}$, thereby leading to the improvement of the approximation accuracy. In such setting, the state and input sets are not rediscretized, and the obtained abstraction is refined.

\section{Controller Synthesis}
\label{sec-controller}

With the feedback refinement relation established in Section \ref{sec-symbolic}, the next step is to study controller synthesis for the system $T_{\tau}(\Sigma)$ via its abstraction $T_{\tau, \eta, \mu}(\Sigma)$. To begin with, we recall the definition of the abstract specification from \cite{Reissig2017feedback}.

\begin{definition}
\label{def-5}
Given a transition system $T:=(X, X^{0}, U, \Delta, Y, H)$ and a set $Z\in\mathbb{R}^{n}$, any subset $\mathcal{S}\subseteq Z^{\infty}$ is called a \textit{specification} on $Z$. The system $T$ is said to \textit{satisfy a specification $\mathcal{S}$} on $U\times Y$ if from certain time instant, the trajectory of $T$ always belongs to $\mathcal{S}$.
\end{definition}

\begin{definition}
\label{def-6}
Given two transition systems $T_{i}:=(X_{i}, X^{0}_{i}, U_{i}, \Delta_{i}, Y_{i}, H_{i})$, $i\in\{1, 2\}$. Let $\mathcal{F}\subseteq X_{1}\times X_{2}$ be a relation and $\mathcal{S}_{1}$ be a specification on $U_{1}\times X_{1}$. A specification $\mathcal{S}_{2}$ on $U_{2}\times X_{2}$ is called an \textit{abstract specification} associated with $T_{1}, T_{2}, \mathcal{S}_{1}$ and $\mathcal{F}$, if $(u, x_{1})\in\mathcal{S}_{1}$ holds for all $(u, x_{2})\in\mathcal{S}_{2}$ and all $x_{1}\in X_{1}$ with $(x_{1}, x_{2})\in\mathcal{F}$.
\end{definition}

If $T_{1}\preceq_{\mathcal{F}}T_{2}$ and $\mathcal{S}_{2}$ is an abstract specification associated with $T_{1}, T_{2}, \mathcal{S}_{1}$ and $\mathcal{F}$, then we write $(T_{1}, \mathcal{S}_{1})\preceq_{\mathcal{F}}(T_{2}, \mathcal{S}_{2})$ for the sake of simplicity. In the following, we assume that $\mathcal{F}$ is a feedback refinement relation. For the control system $T_{\tau}(\Sigma)$, assume that the desired specification is given by $\mathcal{S}:=\bar{U}_{1}\times\bar{X}_{1}\subseteq U_{1}\times X_{1}$ with $\bar{U}_{1}=\bigcup_{x\in \bar{X}_{1}}\enab(x)$. Define
\begin{align*}
\bar{X}_{2}&:=\left\{\hat{q}\in X_{2}: (x, \hat{q})\in\mathcal{F}, x\in\bar{X}_{1}, \hat{q}\subseteq\bar{X}_{1}\right\},  \\
\bar{U}_{2}&:=\bigcup_{q\in[\mathbb{R}^{n}]_{\eta}, \hat{q}\in\bar{X}_{2}}U_{2}(\hat{q}), \quad Q_{\eta}(\mathcal{S}):=\bar{U}_{2}\times\bar{X}_{2}.
\end{align*}
As a result, $Q_{\eta}(\mathcal{S})\subseteq U_{2}\times X_{2}$. 

\begin{proposition}
\label{prop-1}
Assume that $T_{\tau}(\Sigma)\preceq_{\mathcal{F}}T_{\tau, \eta, \mu}(\Sigma)$. If $\mathcal{S}\subseteq U_{1}\times X_{1}$ is a specification for the control system $T_{\tau}(\Sigma)$, then $Q_{\eta}(\mathcal{S})$ is a abstract specification for $T_{\tau, \eta, \mu}(\Sigma)$.
\end{proposition}

\begin{proof}
For any $(u, \hat{q}_{1})\in Q_{\eta}(\mathcal{S})$, we have that $u\in\bar{U}_{2}\subseteq\bar{U}_{1}$ and $\hat{q}_{1}\in\bar{X}_{2}\subseteq X_{2}$. Since $T_{\tau}(\Sigma)\preceq_{\mathcal{F}}T_{\tau, \eta, \mu}(\Sigma)$, there exists $x_{1}\in X_{1}$ such that $(x_{1}, \hat{q}_{1})\in\mathcal{F}$, which implies that $x_{1}\in\hat{q}_{1}$. Thus, we obtain from the definition of $\bar{X}_{2}$ that $x_{1}\in\bar{X}_{2}\subseteq\bar{X}_{1}$, which in turn gives that $(u, x_{1})\in\mathcal{S}$.

Given a $u\in U_{2}(\hat{q}_{1})$, define $x_{2}:=\Delta_{1}(x_{1}, u)$ and $\hat{q}_{2}:=\Delta_{2}(\hat{q}_{1}, u)\in\bar{X}_{2}\subseteq X_{2}$. $(x_{2}, \hat{q}_{2})\in\mathcal{F}$ holds from $T_{\tau}(\Sigma)\preceq_{\mathcal{F}}T_{\tau, \eta, \mu}(\Sigma)$, which thus implies that $x_{2}\in\hat{q}_{2}$. Hence, $x_{2}\in\bar{X}_{1}\subseteq X_{1}$ holds from the definition of $\bar{X}_{2}$, which indicates that $(u, x_{2})\in\mathcal{S}$. By iteration, we deduce that $(u, x_{2})\in\mathcal{S}$ for all $(u, \hat{q}_{2})\in Q_{\eta}(\mathcal{S})$ and all $(x_{1}, \hat{q}_{1})\in\mathcal{F}$.
\end{proof}

In the following, we recall the definition of the controller for the control system $T=(X, X^{0}, U, \Delta, Y, H)$ from \cite{Girard2012controller}.

\begin{definition}
\label{def-7}
Given a transition system $T=(X, X^{0}, U, \Delta, Y, H)$, a \textit{controller} is a map $\mathbb{C}: X\rightarrow2^{U}$, and is \textit{well-defined} if $\mathbb{C}(x)\subseteq\enab(x)$ for all $x\in X$. The \textit{controlled system} is denoted by the transition system $T_{c}=(X, X^{0}, U, \Delta_{c}, Y, H)$ with the transition relation given by $x'\in\Delta_{c}(x, u)$ if and only if $u\in \mathbb{C}(x)$ and $x'\in\Delta(x, u)$.
\end{definition}

According to Proposition \ref{prop-1} and Theorem VI.3 in \cite{Reissig2017feedback}, the following result is direct, and the proof is omitted here.

\begin{proposition}
\label{prop-2}
If $(T_{\tau}(\Sigma), \mathcal{S})\preceq_{\mathcal{F}}(T_{\tau, \eta, \mu}(\Sigma), Q_{\eta}(\mathcal{S}))$ and $\mathbb{C}_{1}: X_{2}\rightarrow2^{U_{2}}$ is a controller for $(T_{\tau, \eta, \mu}(\Sigma), Q_{\eta}(\mathcal{S}))$, then the map $\mathbb{C}: X_{1}\rightarrow2^{U_{1}}$, defined as $\mathbb{C}(x):=\mathbb{C}_{1}(\mathcal{F}(x))$, is a controller for $(T_{\tau}(\Sigma), \mathcal{S})$.
\end{proposition}

\subsection{Safety Controller Synthesis}
\label{subsec-safety}

Let $\mathcal{O}_{s}\subseteq Y$ be a output set associated with safe states. In this subsection, we consider the safety synthesis problem, which is to determine a controller to keep the system output inside the specified safe set $\mathcal{O}_{s}$.

\begin{definition}[see \cite{Girard2013low}]
\label{def-8}
Let $\mathcal{O}_{s}\subseteq Y$ be a set of safe outputs. A controller $\mathbb{C}$ is a \textit{safety controller} for $T_{c}=(X, X^{0}, U, \Delta_{c}, Y, H)$ with the specification $\mathcal{O}_{s}$, if for all $x\in\dom(\mathbb{C})$, (i) $H(x)\in\mathcal{O}_{s}$; (ii) $\forall u\in\mathbb{C}(x)$, $\Delta_{c}(x, u)\subseteq\dom(\mathbb{C})$, where $\dom(\mathbb{C}):=\{x\in X: \mathbb{C}(x)\neq\varnothing\}$.
\end{definition}

\begin{lemma}[see \cite{Girard2012controller}]
\label{lem-1}
Given a transition system $T$ with the specification $\mathcal{O}_{s}$, a controller $\mathbb{C}$ is a safety controller if and only if for all the non-blocking states of the controlled system $T_{c}$, $H(x)\in\mathcal{O}_{s}$ and $x'\in\Delta(x, \mathbb{C}(x))$ is non-blocking.
\end{lemma}

We are now in the position to design a safety controller for the control system $T_{\tau}(\Sigma)$ with the specification $\mathcal{O}_{s}$.

\begin{theorem}
\label{thm-2}
Assume that $T_{\tau}(\Sigma)\preceq_{\mathcal{F}}T_{\tau, \eta, \mu}(\Sigma)$. If $\mathbb{C}_{1}: X_{2}\rightarrow2^{U_{2}}$ is a safety controller for $T_{\tau, \eta, \mu}(\Sigma)$ with the specification $Q_{\eta}(\mathcal{O}_{s})$, let $\mathbb{C}: X_{2}\rightarrow2^{U_{1}}$ be given by
\begin{equation}
\label{eqn-6}
\mathbb{C}(x):=\mathbb{C}_{1}(\mathcal{F}(x)), \quad \forall x\in X_{1},
\end{equation}
then the map $\mathbb{C}: X_{1}\rightarrow2^{U_{1}}$ is well-defined, and is a safety controller for $T_{\tau}(\Sigma)$ with the specification $\mathcal{O}_{s}$.
\end{theorem}

\begin{proof}
First, we prove that the controller $\mathbb{C}$ is well-defined. Let $x_{1}\in X_{1}$, and $u\in\mathbb{C}(x_{1})$. It follows from \eqref{eqn-6} that there exists $\hat{q}_{1}\in X_{2}$ such that $(x_{1}, \hat{q}_{1})\in\mathcal{F}$ and $u\in\mathbb{C}_{1}(\hat{q}_{1})$. Since $\mathbb{C}_{1}$ is well-defined, we have that $u\in\enab(\hat{q}_{1})$. That is, there exists $\hat{q}_{2}\in\Delta_{2}(\hat{q}_{1}, u)$. It follows from the feedback refinement relation that there exists $x_{2}\in\Delta_{1}(x_{1}, u)$ such that $(x_{2}, \hat{q}_{2})\in\mathcal{F}$, which implies that $u\in\enab(x_{1})$. Thus, for all $x_{1}\in X_{1}$, $\mathbb{C}(x_{1})\subseteq\enab(x_{1})$. Thus, $\mathbb{C}$ is well-defined.

Next, we prove that $\mathbb{C}$ is a safety controller for the specification $\mathcal{O}_{s}$. Let $x_{1}\in X_{1}$ such that $\mathbb{C}(x_{1})\neq\varnothing$, and let $u\in\mathbb{C}(x_{1})$. By \eqref{eqn-6}, there exists $\hat{q}_{1}\in X_{2}$ such that $(x_{1}, \hat{q}_{1})\in\mathcal{F}$ and $u\in\mathbb{C}_{1}(\hat{q}_{1})$. Since $\mathbb{C}_{1}$ is a safety controller for the specification $Q_{\eta}(\mathcal{O}_{s})$ and $\mathbb{C}_{1}(x_{1})\neq\varnothing$, we have from Lemma \ref{lem-1} that $\hat{q}_{1}\in Q_{\eta}(\mathcal{O}_{s})$. It follows from Proposition \ref{prop-1} that $Q_{\eta}(\mathcal{O}_{s})$ is an abstraction of the specification $\mathcal{O}_{s}$. Therefore, we obtain from $(x_{1}, \hat{q}_{1})\in\mathcal{F}$ that $x_{1}\in\mathcal{O}_{s}$.

Let $x_{2}:=\Delta_{1}(x_{1}, u)$. We have from $T_{\tau}(\Sigma)\preceq_{\mathcal{F}}T_{\tau, \eta, \mu}(\Sigma)$ that there exists $\hat{q}_{2}\in\Delta_{2}(\hat{q}_{1}, u)$ such that $(x_{2}, \hat{q}_{2})\in\mathcal{F}$. Since $\mathbb{C}_{1}$ is a safety controller for the specification $Q_{\eta}(\mathcal{O}_{s})$ and $u\in\mathbb{C}_{1}(\hat{q}_{1})$, we have, from Lemma \ref{lem-1}, that $\mathbb{C}_{1}(\hat{q}_{2})\neq\varnothing$. Finally, \eqref{eqn-6} implies that $\mathbb{C}_{1}(\hat{q}_{2})\subseteq\mathbb{C}(x_{2})$, and therefore $\mathbb{C}(X_{2})\neq\varnothing$. As a result, we conclude that $\mathbb{C}$ is a safety controller for the specification $\mathcal{O}_{s}$, which completes the proof.
\end{proof}

\begin{algorithm}[!t]
\caption{Safe Controller Design}
\label{alg-1}
\hspace*{0.02in} {\bf Input:} $\Upsilon, \Upsilon_{1}\subseteq X_{2}, \mathbb{C}_{1}$ \\
\hspace*{0.02in} {\bf Output:} $(X_{c}, \mathbb{C}_{1})$ with $X_{c}=\Upsilon$
\begin{algorithmic}[1]
\State Explore($\Upsilon\backslash\Upsilon_{1}$)
\State $W=\cp(\Upsilon)\cap\Upsilon$
\State $\mathbb{C}_{1}=\mathbb{C}_{1}\cup\{(W, U_{2}, \varnothing)\}$
\State $\Upsilon_{1}=\Upsilon_{1}\cup W$
\If{$\Upsilon=\Upsilon_{1}$}
\State \Return $(\Upsilon, \mathbb{C}_{1})$
\Else
\For{$\Upsilon\neq\Upsilon_{1}$}
\State $\Upsilon=\Upsilon_{1}$
\State $\Upsilon_{1}=\varnothing$
\State Explore($\Upsilon$)
\State $W=\cp(\Upsilon)\cap\Upsilon$
\State $\mathbb{C}_{1}=\mathbb{C}_{1}\cup\{(W, U_{2}, \varnothing)\}$
\State $\Upsilon_{1}=W$
\EndFor
\EndIf
\end{algorithmic}
\end{algorithm}

In Theorem \ref{thm-2}, the abstract specification $Q_{\eta}(\mathcal{O}_{s})$ is applied in synthesizing the safety controller, which is different from the results in \cite{Girard2012controller}, where the contraction and expansion of $\mathcal{O}_{s}$ are used. Following Definition \ref{def-7}, the controller for $T_{\tau, \eta, \mu}(\Sigma)$ can be written as a transition system $\mathbb{C}_{1}=(X_{c}, U_{2}, G)$ with the state set $X_{c}\subseteq X_{2}$, the output set $U_{2}$ and the transition relation $G\subseteq X_{c}\times U_{2}$. In practice, we are interested in a bounded set of states $X_{2}$, which implies that the state set $X_{1}$ is also bounded; see \cite{Meyer2018compositional}. Next, we focus on how to obtain the controller on $X_{2}$ via the abstract specification.

To this end, assume that $\hat{\mathcal{O}}_{s}:=\{\hat{q}\in X_{2}|(x, \hat{q})\in\mathcal{F}, x\in X_{1}, \hat{q}\in\mathcal{O}_{s}\}$ is a under-approximation of $\mathcal{O}_{s}\in Y$. From Proposition \ref{prop-1}, it is easy to verify that $\hat{\mathcal{O}}_{s}$ is a abstract specification of $\mathcal{O}_{s}$. Define the following sets (see \cite{Maler1995synthesis})
\begin{align}
\label{eqn-7}
W_{0}&:=\hat{\mathcal{O}}_{s}, \quad
W_{i+1}:=\cp(W_{i})\cap\hat{\mathcal{O}}_{s}, \quad i\in\mathbb{N}^{+}.
\end{align}
In \eqref{eqn-7}, the function $\cp: 2^{X_{2}}\rightarrow2^{X_{2}}$ is the controllable predecessor operator \cite{Maler1995synthesis}, and defined as: for a set $\Upsilon\subseteq X_{2}$,
\begin{align*}
\cp(\Upsilon)&:=\left\{\hat{q}\in X_{2}|\exists u\in U_{2} \text{ such that }\Delta_{2}(\hat{q}, u)\subseteq\Upsilon\right\}.
\end{align*}
The iteration in \eqref{eqn-7} ends when $W_{i}=W_{i+1}$. Since the sequence $\{W_{i}\}$ is monotone over a finite domain, the convergence of such a sequence is guaranteed in finite time; see \cite{Maler1995synthesis}. Assume that the number of the iterations is $M\in\mathbb{N}^{+}$.

\begin{algorithm}[!t]
\caption{Explore}
\label{alg-2}
\hspace*{0.02in} {\bf Input:} $\Upsilon\subseteq X_{2}$ \\
\hspace*{0.02in} {\bf Output:} transition relation $G$
\begin{algorithmic}[1]
\For{$\hat{q}\in\Upsilon$, $u\in U_{2}$}
\If{$\Delta_{2}(\hat{q}, u)$ is not defined}
\State compute $\Delta_{2}(\hat{q}, u)$
\If{$\Delta_{2}(\hat{q}, u)\subseteq\Upsilon$}
\State map $\hat{q}$ to $u$
\EndIf
\EndIf
\EndFor
\end{algorithmic}
\end{algorithm}

According to \eqref{eqn-7}, define the state set $X_{c}:=W_{M}$. Following the definition of the abstraction $T_{\tau, \eta, \mu}(\Sigma)$, the transition $G: X_{c}\rightarrow U_{2}$ is defined as: for all $q\in X_{c}$ and $u\in U_{2}$, $u=G(\hat{q})\Leftrightarrow\Delta_{2}(\hat{q}, u)\subseteq X_{c}$. As a result, we have that $\mathbb{C}_{1}=(X_{c}, U_{2}, G)$, which is a safety controller for the abstraction $T_{\tau, \eta, \mu}(\Sigma)$ with the specification $\hat{\mathcal{O}}_{s}$.

Based on the above analysis, the synthesis algorithm is summarized in Algorithm \ref{alg-1}, which is terminated after $M$ iterations. In Algorithm \ref{alg-1}, $\Upsilon$ and $\Upsilon_{1}$ are initialized as $\hat{\mathcal{O}}_{s}$ and $\varnothing$, respectively. Algorithm \ref{alg-2} is to develop the maps from $\hat{q}\in\Upsilon$ to $u\in U_{2}$ such that the transition $G$ is obtained. Because of the monotonic nature of the iterative computation of safe sets, the set $\Upsilon$ is a subset of $\hat{\mathcal{O}}_{s}$. Based on \eqref{eqn-7} and Algorithm \ref{alg-1}, the controller $\mathbb{C}_{1}$ is obtained iteratively, which further leads to the controller for the system $T_{\tau}(\Sigma)$.

\begin{theorem}
\label{thm-3}
Assume that $T_{\tau}(\Sigma)\preceq_{\mathcal{F}}T_{\tau, \eta, \mu}(\Sigma)$. The controller $\mathbb{C}_{1}$ obtained via Algorithm \ref{alg-1} is a safety controller for $T_{\tau, \eta, \mu}(\Sigma)$ with the specification $\hat{\mathcal{O}}_{s}$. Furthermore, the controller $\mathbb{C}(x):=\mathbb{C}_{1}(\mathcal{F}(x))$ is a safety controller for $T_{\tau}(\Sigma)$ with the specification $\mathcal{O}_{s}$.
\end{theorem}

\begin{figure}[!t]
\begin{center}
\begin{picture}(80, 85)
\put(-40,-15){\resizebox{50mm}{35mm}{\includegraphics[width=2.5in]{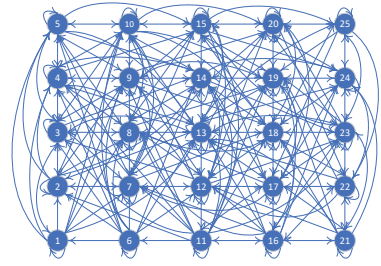}}}
\end{picture}
\end{center}
\caption{Symbolic model $T_{0.2, 0.2}(\Sigma)$ for the control system $\Sigma$. The abstract state $(z_{i}, z_{j})$ in $T_{0.2, 0.2}(\Sigma)$ with $i, j\in\{-2, -1, 0, 1, 2\}$ corresponds to
the state $5(i+2)+j+3$ in this figure.}
\label{fig-2}
\end{figure}

\section{Illustrative Example}
\label{sec-example}

As a simple mechanical control system studied in the literature \cite{Pola2008approximately}, the pendulum is described as
\begin{align*}
\Sigma:  \dot{x}_{1}&=x_{2}, \quad \dot{x}_{2}=-gl^{-1}\sin(x_{1})-km^{-1}x_{2}+u,
\end{align*}
where $x_{1}$ and $x_{2}$ are respectively the angular position and velocity of the point mass, $u$ is the torque which can be treated as the control variable. In addition, $g = 9.8$ is the gravity acceleration, $l=5$ is the length of the rod, $m =0.5$ is the mass, and $k=3$ is the coefficient of friction. Assume that the control input $u$ is piecewise-constant and bounded in the set $U=[-2.5, 2.5]$. For simplicity the state set is bounded in the set $X=[-1, 1]\times[-1, 1]$.

To construct the abstraction, the applied quantizer is
\begin{align*}
Q(z):=\left\{\begin{aligned}
&\frac{(1+\eta)^{k+1}a}{(1-\eta)^{k}},  & & \frac{(1+\eta)^{k}a}{(1-\eta)^{k}}<z\leq\frac{(1+\eta)^{k+1}a}{(1-\eta)^{k+1}}; \\
&0, & & 0\leq z\leq a;  \\
& -Q(-z), & & z<0.
\end{aligned}\right.
\end{align*}
Let $\eta=0.2$ and $a=0.4$, and thus there are 25 quantization regions for the logarithmic quantizer. Comparing with the uniform quantizer applied in \cite{Pola2008approximately}, the quantization regions for the logarithmic quantizer are not the same with equivalent size. By adjusting the parameters $\eta$ and $a$, we can change the precision of the logarithmic quantizer, whereas the precision of the uniform quantizer depends on the precision of the approximate bisimulation; see also \cite{Girard2010approximately, Pola2008approximately}.

\begin{figure}[!t]
\begin{center}
\begin{picture}(80, 80)
\put(-40,-15){\resizebox{60mm}{36mm}{\includegraphics[width=2.5in]{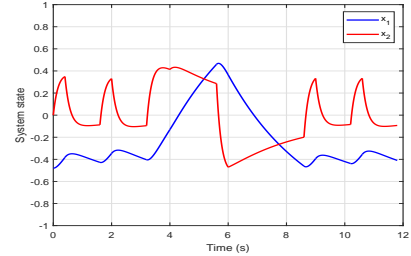}}}
\end{picture}
\end{center}
\caption{Trajectory of the control system $\Sigma$ with initial condition $(-0.48, 0)$ and control strategy synthesized on $T_{0.2, 0.2}(\Sigma)$.}
\label{fig-3}
\end{figure}

\begin{figure}[!t]
\begin{center}
\begin{picture}(80, 90)
\put(-40,-15){\resizebox{60mm}{36mm}{\includegraphics[width=2.5in]{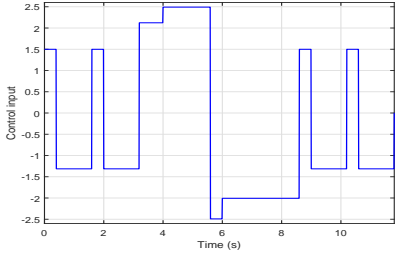}}}
\end{picture}
\end{center}
\caption{Control strategy synthesized on $T_{0.2, 0.2, 2\times10^{-3}}(\Sigma)$.}
\label{fig-4}
\end{figure}

Let $\tau=0.2$ and $\mu=2\times10^{-3}$. In addition, the Lipschitz constant for $\Sigma$ is 6. The symbolic model $T_{0.2, 0.2, 2\times10^{-3}}(\Sigma)=(X_{2}, X^{0}_{2}, U_{2}, \Delta_{2}, Y_{2})$ is given by: (i) $X_{2}$ is the union of the quantization regions partitioned via the logarithmic quantizer $Q$; (ii) $X^{0}_{2}=X_{2}$; (iii) $U_{2}=\bigcup_{q\in [\mathbb{R}^{n}]_{2\times10^{-3}}}U_{2}(q)$; (iv) the transition relation $\Delta_{2}$ is depicted in Fig. \ref{fig-2}; (v) $Y_{2}=X_{2}$. The transition system $T_{0.2, 0.2, 2\times10^{-3}}(\Sigma)$ is shown in Fig. \ref{fig-2}, where the transition relation is obtained via \eqref{eqn-6} and the numerical integration of the trajectories of $\Sigma$. Comparing with the uniform quantization based abstraction in \cite{Pola2008approximately}, there are more (loop) transitions in $T_{0.2, 0.2, 2\times10^{-3}}(\Sigma)$ emanating from the abstract states, which implies that some complexity issues can be avoided; see \cite{Reissig2014feedback}.

In the following, the controller synthesis is illustrated via the symbolic model $T_{0.2, 0.2, 2\times10^{-3}}(\Sigma)$. Assume that the objective is to design a controller to enforce an alternation between two different periodic motions, which are respectively denoted as $\mathcal{S}_{1}$ and $\mathcal{S}_{2}$. The periodic motion $\mathcal{S}_{1}$ requires the state of $\Sigma$ to cycle between $(-0.48, 0)$ and $(0, 0)$, whereas the periodic motion $\mathcal{S}_{2}$ requires the state to cycle between $(-0.48, 0)$ and $(0.48, 0)$. Thus, the control aim is to design a controller such that the system $\Sigma$ satisfies a specification $\mathcal{S}$ requiring the execution of the sequence of periodic motions $\mathcal{S}_{1}, \mathcal{S}_{1}, \mathcal{S}_{2}, \mathcal{S}_{1}, \mathcal{S}_{1}$.

A control strategy for periodic motions $\mathcal{S}_{1}$ and $\mathcal{S}_{2}$ can be obtained by performing a search on $T_{0.2, 0.2, 2\times10^{-3}}(\Sigma)$ using standard methods in supervisory control \cite{Ramadge1987supervisory}. A possible solution for $\mathcal{S}_{1}$ is given by $(-0.48, 0)\overset{1.4991}{\longrightarrow}(0, 0)\overset{-1.2127}{\longrightarrow}(-0.48, 0)$, and a solution for $\mathcal{S}_{2}$ is given by $(-0.48, 0)\overset{2.1230}{\longrightarrow}(0, 0.48)\overset{2.4914}{\longrightarrow}(0.48, 0)\overset{-2.4914}{\longrightarrow}(0, -0.48)\overset{-2.0074}{\longrightarrow}(-0.48, 0)$. Based on such two solutions, a control strategy for $\mathcal{S}$ is derived by combining the trajectories associated with the motions $\mathcal{S}_{1}, \mathcal{S}_{1}, \mathcal{S}_{2}, \mathcal{S}_{1}$ and $\mathcal{S}_{1}$. As a result, we have the following transitions: $(-0.48, 0)\overset{1.4991}{\longrightarrow}(0, 0)\overset{-1.2127}{\longrightarrow}(-0.48, 0)\overset{1.4991}{\longrightarrow}(0, 0)\overset{-1.2127}{\longrightarrow}(-0.48, 0)\overset{2.1230}{\longrightarrow}(0, 0.48)\overset{2.4914}{\longrightarrow}(0.48, 0)\overset{-2.4914}{\longrightarrow}(0, -0.48)\overset{-2.0074}{\longrightarrow}(-0.48, 0)\overset{1.4991}{\longrightarrow}(0, 0)\overset{-1.2127}{\longrightarrow}(-0.48, 0)$. Note that some transitions are not obtained by one sampling period. This means that the abstract state may stay the same after certain transitions, which results from the symbolic abstraction via the logarithmic quantization; see also the loop transitions in Fig. \ref{fig-2}. The control strategy is presented in Fig. \ref{fig-4}. With such control strategy, the evolution of the system state is shown in Fig. \ref{fig-3}. The completion time of the specification $\mathcal{S}$ is 11.8s, whereas the completion time in \cite{Pola2008approximately} is 24s, which implies that the computation time is reduced significantly.

\section{Conclusion}
\label{sec-conclusion}

In this paper, we applied logarithmic quantization to construct the symbolic abstraction for nonlinear control systems. Based on the constructed discrete abstraction, the controller synthesis problem was studied via abstract specification, and a novel algorithm was proposed to compute the safety controller. Finally, a numerical example was provided to illustrate the obtained results. Future researches will be directed to symbolic abstractions via dynamic quantizers.


\end{document}